\documentclass[nolinenumbers]{jfm}

\usepackage{graphicx}
\usepackage{newtxtext}
\usepackage{newtxmath}
\usepackage{natbib}
\usepackage{hyperref}
\hypersetup{
    colorlinks = true,
    urlcolor   = blue,
    citecolor  = black,
}

\usepackage{amsmath}
\usepackage{amssymb}
\newtheorem{proposition}{Proposition}[section]

\newcommand{\RomanNumeralCaps}[1]
\linenumbers

\title{Extending convolutional neural networks to finite-volume meshes for fine-scale reconstruction of wall-bounded turbulence}

\author{Hesam Tofighian\aff{1,2}
  \corresp{\email{hesam.tofighian@kit.edu}},
  Jordan A. Denev\aff{1}
 \and Nikolai Kornev\aff{2}}

\affiliation{\aff{1}Scientific Computing Center, Department of Scientific Computing and Mathematics, Karlsruhe Institute of Technology, Karlsruhe 76131, Germany
\aff{2}Chair of Modelling and Simulation, University of Rostock, Rostock 18051, Germany}

\begin{document}
\maketitle

\begin{abstract}
Standard convolutional neural networks for turbulence super-resolution are restricted to uniform structured grids, preventing their direct application to the non-uniform and non-orthogonal meshes commonly used in wall-bounded flows and complex geometries. We introduce a finite-volume convolution unit (FVCU) that removes this restriction. The key idea is to construct convolutional layers using finite-volume operators that are applicable to non-uniform and unstructured meshes. To this end, we use an exact representation of a 3D convolution as a linear combination of successive axis-wise convolutions. We adopt axis-wise gradient, averaging, and identity operators, because they can be evaluated on arbitrary polyhedral cells using finite-volume discretization schemes. The resulting FVCU is well-defined on general polyhedral meshes, serving as a drop-in replacement for conventional convolution layers. The FVCU is embedded in a super-resolution neural network and trained to reconstruct three-dimensional velocity fields from $4\times$ coarsened low-resolution fields in turbulent channel flow on a non-uniform orthogonal mesh, and in circular-pipe flow on a non-orthogonal O-grid mesh. Inspired by the role of the Van Driest damping function in near-wall subgrid-scale modelling in large-eddy simulations, this function is supplied as a conditional input to encode wall proximity. An ablation study confirms that the network interprets this input in a physically consistent manner. In both configurations, the model accurately reconstructs instantaneous velocity fields, Reynolds-stress profiles and their anisotropy, and energy spectra beyond the Nyquist limit of the low-resolution mesh. On an elliptic pipe unseen during training, the model reproduces the main flow features and local statistics, with a modest overprediction of small-scale energy. These results demonstrate a practical framework for extending CNN-based turbulent-flow super-resolution to general finite-volume meshes.
\end{abstract}


\section{Introduction}\label{sec:intro}
Turbulent flows are governed by the interaction of vortical structures whose sizes range from the large energy-containing motions down to the finest dissipative structures. Resolving this full hierarchy through direct numerical simulation (DNS) demands grids fine enough to capture the smallest motions everywhere, at a cost that grows with Reynolds number \citep{pope2001turbulent}. Large-eddy simulation alleviates this problem by filtering out structures below the grid cutoff, but the resulting fields lack the small-scale content on which processes such as dissipation, particle dispersion, and near-wall momentum transfer depend. Experimental measurement techniques such as particle image velocimetry (PIV) are similarly constrained by their finite spatial resolution, which can significantly affect the accuracy of turbulence statistics, and obtaining finer resolution demands more expensive equipment \citep{ayegba2020turbulence}. The distribution of turbulent length scales is generally not spatially uniform. In wall-bounded flows, for instance, the characteristic size of turbulent structures decreases significantly toward the wall, reaching length scales in the viscous sublayer that are orders of magnitude smaller than in the bulk flow. Practical computational grids reflect the spatial non-uniformity of the relevant scales by clustering cells where small structures concentrate and stretching them elsewhere, producing meshes that are inherently non-uniform and, in complex geometries, often non-orthogonal. Obtaining high-resolution (HR) turbulence data on such configurations therefore remains a challenge in fluid dynamics.

A fundamentally similar challenge, namely the recovering fine-scale spatial detail from low-resolution (LR) data, has been extensively studied in computer vision through image super-resolution techniques. \citet{dong2015image} first demonstrated that a convolutional neural network (CNN) trained to map LR images directly to their HR counterparts could outperform classical interpolation methods. Reconstruction quality has since been progressively improved through deeper residual architectures \citep{kim2016accurate} and the reuse of intermediate features via residual dense blocks \citep{zhang2018residual} and their nested variant, the residual-in-residual dense block \citep{wang2018esrgan}. Further improvements in super-resolution realism have been achieved by incorporating generative adversarial training \citep{ledig2017photo}, transformer architectures \citep{liang2021swinir}, and more recently, employing diffusion-based generative models \citep{saharia2022image}.

Inspired by these developments in image super-resolution, data-driven models have been increasingly adopted to reconstruct HR turbulent flow fields from coarse simulation or experimental data. However, because these models inherit the standard convolution operator from image processing, they are restricted to the cases with uniform, structured grids. \citet{fukami2019super} pioneered the application of super-resolution techniques to turbulent flows, using machine learning models based on CNNs to reconstruct two-dimensional (2D) decaying turbulent velocity and vorticity fields from extremely coarse input data. \citet{liu2020deep} extended this framework and showed that incorporating consecutive LR snapshots improves reconstruction accuracy. Their CNN model was trained on 2D slices from forced isotropic turbulence and turbulent channel flow. For the latter, because standard convolutional operators require uniform grid spacing, the model was restricted to wall-parallel planes on which the streamwise and spanwise grids are equidistant. Planes were sampled across the viscous sublayer, buffer layer, log-law region, and outer layer, exposing the network to turbulence at different wall distances during training. \citet{fukami2021machine} extended this line of work to spatio-temporal super resolution, combining spatial reconstruction with the recovery of intermediate snapshots between sparsely sampled time frames. To apply three-dimensional (3D) convolutions to turbulent channel flow data, the originally non-uniform wall-normal mesh was first interpolated onto a uniform grid prior to training. \citet{zhou2022robust} also employed 3D convolutions for volumetric super resolution of homogeneous isotropic turbulence and introduced a patchwise training strategy in which the model is trained on small patches extracted from DNS at multiple Reynolds numbers, enabling it to generalize to higher Reynolds numbers not seen during training. Beyond purely supervised approaches, \citet{deng2019super} employed generative adversarial networks to reconstruct turbulent velocity fields from PIV measurement images, and \citet{bode2021using} further incorporated physics-informed loss terms, such as a continuity residual, into the adversarial training framework to improve the physical consistency of the reconstructed fields. Unsupervised learning has also been explored to remove the need for paired training data. \citet{kim2021unsupervised} employed a cycle-consistent generative adversarial network (CycleGAN) for super-resolution reconstruction of turbulent channel flow from unpaired data, and \citet{maejima2025unsupervised} extended the approach to subgrid-scale modelling for very coarse-grid large-eddy simulations of channel flow. Both studies were similarly restricted to wall-parallel planes by the uniform-grid requirement of standard convolutions. To account for the variation of turbulence across the boundary layer, \citet{kim2021unsupervised} trained separate networks at individual wall-normal locations, while \citet{maejima2025unsupervised} adopted a strategy similar to \citet{liu2020deep}, training a single network on slices sampled at multiple wall-normal distances. All of the above studies rely on standard convolutional operators that assume uniform grid spacing. This assumption is naturally satisfied in cases such as homogeneous isotropic turbulence, image-based PIV data, and wall-parallel planes of channel flows, but it excludes data on non-uniform and non-orthogonal meshes.

Only a few recent studies have addressed this restriction by designing super-resolution architectures for more general grid configurations. \citet{du2024conditional} employed the SIREN network architecture \citep{sitzmann2020implicit} to represent flow fields as continuous functions of spatial coordinates (neural fields), and applied a diffusion-based generative model in the resulting latent space rather than directly on mesh nodes. Since the neural field can be queried at arbitrary locations, the framework naturally accommodates unstructured grids. Turbulent flow super-resolution was demonstrated as one of conditional generation tasks of their method. However, because the learned representation is tied to the specific training case, a new model must be trained for each flow configuration and geometry. \citet{fukami2023super} suggested graph neural networks (GNNs) as a promising approach for developing super-resolution models on unstructured grids. Following this direction, \citet{barwey2025mesh} proposed a GNN model based on the message-passing framework of MeshGraphNets \citep{pfaff2020learning} for super-resolution on spectral-element grids. In this method, message passing layers operate locally on graph representations of each hexahedral element to recover fine-scale flow fields by increasing the polynomial order of the element. They demonstrated the method on Taylor–Green vortex and backward-facing step flows and examined its generalization to a cavity flow.

In the present work, we explore an alternative route to super-resolution on general meshes by reformulating the convolution operator while retaining established CNN architectures. To this end, we propose a CNN-based architecture that, unlike standard CNNs, can operate directly on unstructured finite-volume meshes by expressing 3D convolution kernels as linear combinations of compositions of axis-wise operators. We show that this decomposition admits an exact representation on uniform Cartesian grids in terms of gradient, averaging, and identity operators, which we then extend to general meshes using standard finite-volume discretizations defined on arbitrary polyhedral cells. The resulting finite-volume convolution unit (FVCU) serves as a replacement for the standard convolution layer and can be embedded in established deep-learning architectures. We demonstrate the method on super-resolution of turbulent channel and circular-pipe flows, which provide complementary assessments on non-uniform orthogonal and non-orthogonal meshes, respectively. Generalization to an unseen geometry is further assessed using turbulent elliptic-pipe flow. Inspired by the conditional input strategy of \citet{tofighian2024conditional} and \citet{shamooni2025super}, who conditioned the generator on subgrid-scale information, we additionally supply the Van Driest damping function as a wall-proximity encoding. This extra input channel provides the network with explicit information about wall proximity to support the reconstruction of strongly anisotropic turbulence in wall-bounded flows.

The remainder of the paper is organized as follows. Section~\ref{sec:method} presents the axis-wise decomposition of 3D convolutions, introduces the finite-volume basis operators, and describes the resulting network architecture along with the numerical setup used for data generation and training. Results are presented in Section~\ref{sec:results}, where the reconstruction accuracy is assessed through instantaneous fields, turbulence statistics, and spectral analysis, followed by an ablation study on the wall-proximity encoding. Conclusions are drawn in Section~\ref{sec:Conclusion}.

\section{Methodology}\label{sec:method}
Standard CNNs assume uniform grid spacing, which prevents their direct application to the non-uniform and unstructured computational grid commonly used in finite-volume simulations. To overcome this limitation and enable CNN-based super-resolution on general computational grids, we propose a finite-volume-based convolutional unit by expressing 3D convolution operators as linear combinations of axis-wise operators. By choosing axis-wise operators that can be defined within the finite-volume framework, convolutional layers can be implemented using operators that remain well-defined on non-uniform and unstructured grids.

\subsection{Axis-wise representation of 3D convolutions}
Let $X:\mathbb{Z}^3\to\mathbb{R}$ denote a discrete feature map and let $K\in\mathbb{R}^{3\times3\times3}$ be a convolution kernel. We denote by $\star$ the standard 3D convolution, and by $\star_x$, $\star_y$, and $\star_z$ the axis-wise convolutions along the corresponding coordinate directions. For a fixed basis $\{b_1,b_2,b_3\}\subset\mathbb{R}^3$, the 3D convolution $K\star X_{in}$ admits the representation
\begin{equation}
\label{eq:main_method}
K\star X_{in}
=
\sum_{p,q,r=1}^{3}
c_{pqr}\,
\Bigl(
b_r\star_z\bigl(b_q\star_y(b_p\star_x X_{in})\bigr)
\Bigr).
\end{equation}
This identity is exact; existence and uniqueness of the coefficients $c_{pqr}$ are established in Appendix~A.

A direct computational realization of this representation is illustrated in Fig. \ref{fig:method}(a). Three axis-wise filters are applied sequentially along the $x$-, $y$-, and $z$-directions, producing $27 C_{in}$ intermediate feature maps that are combined linearly to recover $X_{out}=K\star X_{in}$. Here, $C_{in}$ denotes the number of input channels, and the factor $27=3^3$ arises from applying three basis operators in each of the three coordinate directions. Although this realization is exact, it requires simultaneous storage of all $27C_{in}$ intermediate feature maps and therefore results in a large peak memory footprint.

The same operator can be evaluated with a reduced peak memory footprint by reordering the summation
in Eq. \eqref{eq:main_method} as
\begin{equation}
\label{eq:reordered}
K\star X_{in}
=
\sum_{r=1}^{3}
b_r\star_z\left(
\sum_{p,q=1}^{3}
c_{pqr}\,
\bigl(b_q\star_y(b_p\star_x X_{in})\bigr)
\right).
\end{equation}
In this form, the linear combination over the $x$- and $y$-directional responses is evaluated prior to the final convolution in the $z$-direction. The corresponding realization, shown in Fig. \ref{fig:method}(b), requires storage of only the $9C_{in}$ intermediate feature maps associated with the $x$- and $y$-filter banks. The reduction in peak memory follows solely from the evaluation order; the two realizations define the same convolution operator. We note that the ordering of the axis-wise convolutions in Eqs.~\eqref{eq:main_method} and \eqref{eq:reordered} is arbitrary; any permutation of $x$, $y$, and $z$ yields the same operator.

\begin{figure}
  \centering
  \includegraphics[width=0.75\textwidth]{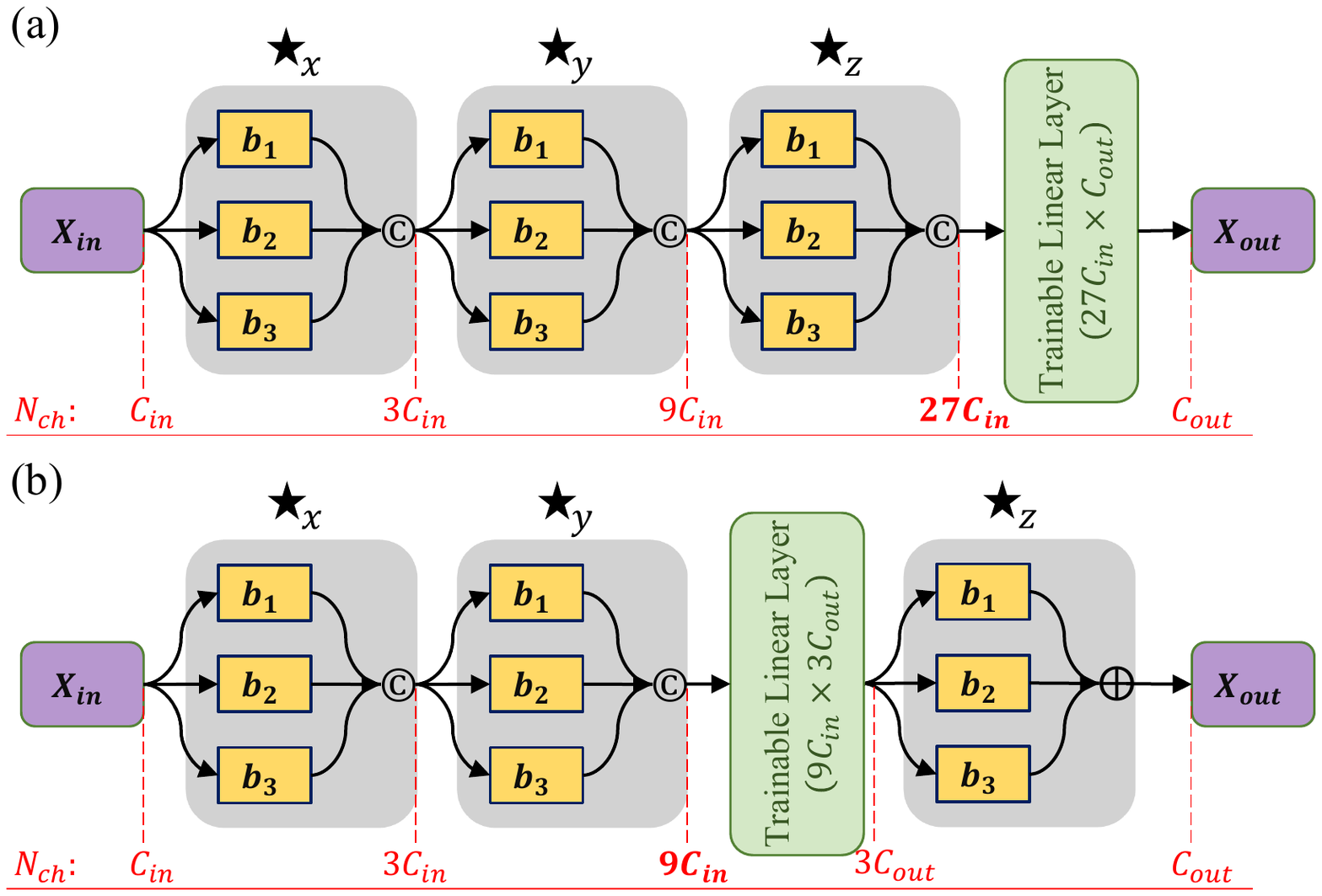}
  \caption{Axis-wise realizations of the 3D convolution operator defined in Eq.~\eqref{eq:main_method} and \eqref{eq:reordered}. Fixed axis-wise basis operators ${b_1,b_2,b_3}$ are applied along the $x$-, $y$-, and $z$-directions (grey blocks). \textcircled{c} denotes channel-wise concatenation and $\oplus$ denotes summation. (a) Direct realization of Eq.~\eqref{eq:main_method}, producing $27C_{\mathrm{in}}$ intermediate feature maps combined by a trainable linear layer. (b) Memory-efficient realization corresponding to Eq.~\eqref{eq:reordered}, in which the linear combination of $x$- and $y$-directional responses precedes the final $z$-directional convolution, reducing peak storage to $9C_{\mathrm{in}}$. Both realizations define the same operator.}
\label{fig:method}
\end{figure}

\subsection{Finite-volume convolution unit (FVCU)}
\label{fvcu}
The decomposition in Eq.~\eqref{eq:reordered} holds for any basis $\{b_1,b_2,b_3\}$ of $\mathbb{R}^3$ and therefore leaves the choice of the axis-wise operators open. We first fix this basis among the finite-difference operators on a uniform grid, and then replace each operator by its finite-volume discretization, so that the resulting layer remains well defined on non-uniform and unstructured meshes.

Consider a uniform Cartesian grid with unit cell size, where $X_{i-1}$, $X_i$, and $X_{i+1}$ are the values at a node and its two $x$-neighbours. We choose the identity, gradient, and averaging operators as the axis-wise basis, which on a uniform grid take the second-order forms
\begin{equation}
\label{eq:fd-ops}
\begin{aligned}
\mathcal I[X]_i &= X_i, &&\qquad b_1 = [\,0\;\;\;1\;\;\;0\,],\\[2pt]
\mathcal G_x[X]_i &= \frac{X_{i+1}-X_{i-1}}{2},
&&\qquad b_2 = \tfrac12\,[\,{-1}\;\;\;0\;\;\;1\,],\\[2pt]
\mathcal A_x[X]_i &= \frac{X_{i-1}+2X_i+X_{i+1}}{4},
&&\qquad b_3 = \tfrac14\,[\,1\;\;\;2\;\;\;1\,],
\end{aligned}
\end{equation}
where $b_1$, $b_2$, and $b_3$ are the corresponding coefficient vectors in $\mathbb{R}^3$, and analogous definitions apply along the $y$- and $z$-directions. These three vectors are linearly independent and therefore form a basis, so based on Eq.~\eqref{eq:reordered}, any kernel $K\in\mathbb{R}^{3\times3\times3}$ is recovered by an appropriate setting of the trainable coefficients $\{c_{pqr}\}$ (Appendix~\ref{appA}). For a cell size $h\neq1$ the gradient operator acquires a factor $1/h$, which leaves the three vectors linearly independent and is absorbed during training.

On a general mesh the finite-difference forms in Eq.~\eqref{eq:fd-ops} are no longer available, since they rely on a structured arrangement of neighbouring cells, but each of the three operators has a well-established finite-volume discretization on arbitrary polyhedral cells. Let $X$ be a cell-centered feature map defined on a finite-volume mesh, and let $P$ be a cell with volume $V_P$ and face set $\mathcal F(P)$. Each face $f\in\mathcal F(P)$ has area $A_f$ and outward unit normal $\mathbf n_f=(n_{f,x},n_{f,y},n_{f,z})$, and $X_f$ is the face-interpolated value of $X$. The identity operator is unchanged, $\mathcal I[X]_P = X_P$, while the gradient operator along $\mathbf e_x$ is given by the finite-volume approximation

\begin{equation}
\mathcal G_x[X]_P
=
(\nabla X)_P \cdot \mathbf e_x
\;\approx\;
\frac{1}{V_P}
\sum_{f\in\mathcal F(P)}
X_f\, n_{f,x}\, A_f.
\label{eq:ufvm-gradx}
\end{equation}
The averaging operator along the same direction follows from the non-negative
weights
\begin{equation}
w_{f,x} = A_f\,|n_{f,x}|,
\qquad
W_x = \sum_{f\in\mathcal F(P)} w_{f,x},
\end{equation}
which give the $x$-directional averaging operator
\begin{equation}
\mathcal A_x[X]_P
=
\frac{1}{W_x}
\sum_{f\in\mathcal F(P)}
w_{f,x}\,X_f,
\label{eq:ufvm-avgx}
\end{equation}
a weighted surface average over faces whose normals have a non-zero projection onto the $x$-direction.

The cell-face values $X_f$ required by these operators are obtained by second-order linear interpolation. At boundary faces, boundary conditions are applied directly; periodic matching is enforced on periodic boundaries, and Dirichlet boundary conditions with zero value are applied on wall boundaries. The Dirichlet boundary condition is analogous to the zero-padding commonly used in CNNs, which has been shown to be robust in practice. On a uniform Cartesian mesh with unit cell size, Eqs.~\eqref{eq:ufvm-gradx} and \eqref{eq:ufvm-avgx} reduce to the operators in Eq.~\eqref{eq:fd-ops}.

We encapsulate this realization of Eq.~\eqref{eq:reordered} in a block, termed the \emph{finite-volume convolution unit} (FVCU), illustrated in Fig.~\ref{fig:fvcu}. The unit applies the fixed operators $\{\mathcal I,\mathcal G_x,\mathcal A_x\}$ followed by $\{\mathcal I,\mathcal G_y,\mathcal A_y\}$, concatenates the resulting $9C_{\mathrm{in}}$ directional responses, and combines them through a pointwise trainable linear map corresponding to the coefficients $\{c_{pqr}\}$. The mixed features are then processed by the $z$-directional operators $\{\mathcal I,\mathcal G_z,\mathcal A_z\}$ and summed to produce the output feature maps.

\begin{figure}
  \centering
  \includegraphics[width=0.6\textwidth]{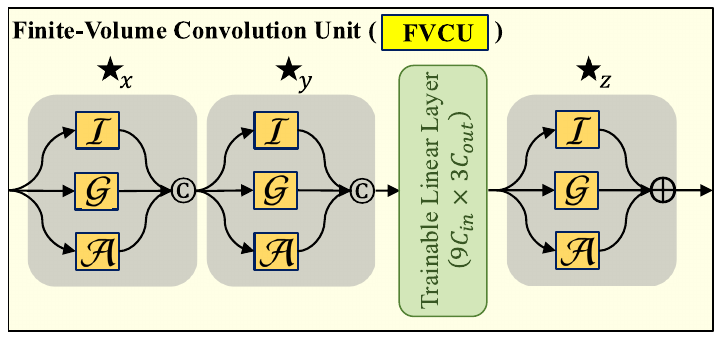}
  \caption{Finite-volume convolution unit (FVCU). Fixed basis operators ($\mathcal{I}$, $\mathcal{G}$, $\mathcal{A}$) are applied sequentially along the $x$-, $y$-, and $z$-directions. \textcircled{c}  denotes channel-wise concatenation, and $\oplus$ denotes summation. The trainable linear layer corresponds to the coefficients $\{c_{pqr}\}$ in Eq.~\eqref{eq:reordered}.}
\label{fig:fvcu}
\end{figure}

\subsection{Super-resolution network architecture}
\label{sec:sr-architecture}
This work employs the residual-in-residual dense architecture of \citet{wang2018esrgan} for super-resolution. We replace the standard convolution layers with the FVCU introduced in Sec.~\ref{fvcu}. The resulting network is referred to as the finite-volume convolutional neural network (FVCNN) hereafter. The network maps a LR cell-centered velocity field on a coarse finite-volume mesh to a super-resolved (SR) velocity field on a refined finite-volume mesh. The network takes four input channels: the normalized LR velocity components \((u^{LR},v^{LR},w^{LR})\) and the Van Driest damping function \(D\) \citep{van1956turbulent},
\begin{equation}
D = 1 - e^{-y^+/A^+}.
\label{eq:vandriest}
\end{equation}
where $y^+$ is the wall-normal distance in inner units and $A^+ = 26$ is the damping constant. The Van Driest damping function is widely used in large-eddy simulation to model the attenuation of turbulence in the near-wall region \citep{pope2001turbulent}. In the present work, it is used as an additional conditional input to encode wall proximity, which is essential in wall-bounded flows where turbulent structures and anisotropy depend strongly on distance from the wall. Expressing this information through \(D\), rather than \(y^+\), maps the unbounded wall coordinate onto a fixed range \([0,1]\), thereby facilitating learning and generalization of near-wall turbulence features. The network outputs the three normalized SR velocity components \((u^{SR},v^{SR},w^{SR})\)
on the associated refined mesh.

The architecture is summarized in Fig.~\ref{fig:architecture}. An initial FVCU maps the four-channel input to \(64\) feature channels. The network trunk consists of a single residual-in-residual dense block (RRDB), formed by three residual dense blocks (RDBs) connected in series. Each RDB contains five FVCU layers, in which intermediate features are concatenated along the channel dimension so that successive layers operate on the aggregation of preceding responses. Nonlinearity is introduced through Leaky Rectified Linear Unit (LReLU) activation functions. The output of each RDB is combined with its input through a residual connection with scaling factor $\beta=0.2$. At the RRDB level, the stack of three RDBs is wrapped by an additional residual connection, again with scaling factor $\beta=0.2$. A long skip connection further adds the RRDB output back to the shallow features, preserving large-scale content while allowing the trunk to focus on super-resolution corrections. 
\begin{figure}
  \centering
  \includegraphics[width=1.0\textwidth]{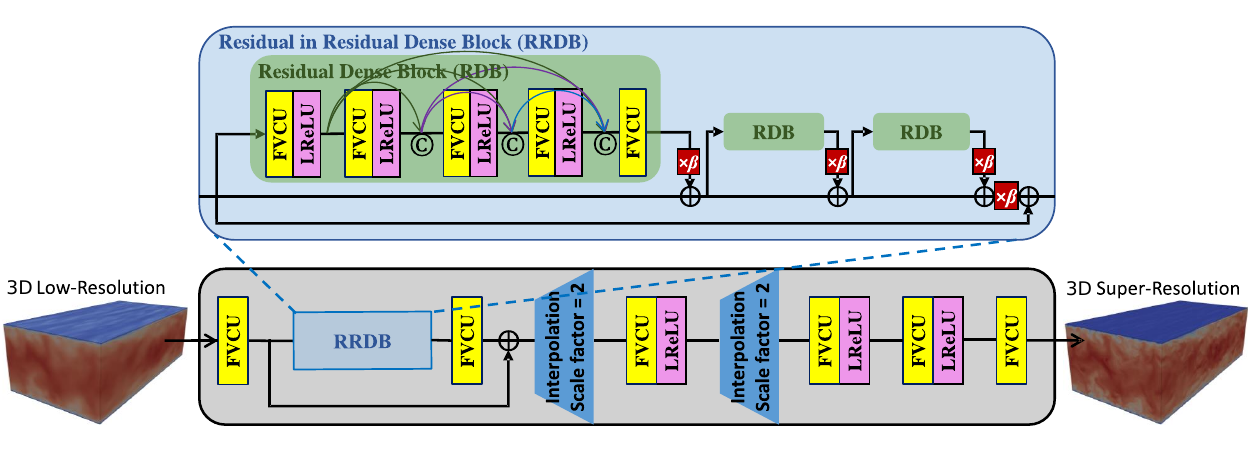}
  \caption{Super-resolution network composed of residual-in-residual dense blocks (RRDBs) built from FVCUs. \textcircled{c} denotes channel-wise concatenation and $\oplus$ denotes summation.}
\label{fig:architecture}
\end{figure}
Spatial refinement is performed in two stages, corresponding to \(2\times\) and \(4\times\) resolution increases. Features are interpolated from the coarse mesh to finer meshes at each stage and subsequently processed by additional FVCU layers to adapt them to the local geometry of the refined mesh. One FVCU is applied after the \(2\times\) refinement, followed by three FVCUs on the \(4\times\) mesh, with the final layer mapping the features to the normalized SR velocity components.

\subsection{Numerical simulation and data preparation}
\label{numeric}
Supervised super-resolution learning requires paired LR velocity fields and corresponding HR ground-truth data. In the present work, the training data are obtained from direct numerical simulations (DNS) of two canonical wall-bounded flows, turbulent plane channel flow and turbulent circular-pipe flow. The two configurations are complementary, since the channel mesh is orthogonal but non-uniform in the wall-normal direction, whereas the circular pipe mesh is both non-uniform and non-orthogonal, so that neither case is accessible to a standard 3D convolution. All three flows are governed by the incompressible Navier--Stokes equations,

\begin{equation}
\nabla\cdot\mathbf{u}=0,
\label{eq:cont}
\end{equation}
\begin{equation}
\frac{\partial\mathbf{u}}{\partial t}
+\mathbf{u}\cdot\nabla\mathbf{u}
=
-\nabla p
+\nu\nabla^2\mathbf{u}
+\mathbf{f},
\label{eq:mom}
\end{equation}

where $\mathbf{u}=(u,v,w)$ denotes the velocity vector, $p$ is the kinematic pressure, $\nu$ is the kinematic viscosity, and $\mathbf{f}$ is a body-force term applied in the streamwise direction to maintain a constant bulk flow rate.

For the channel flow, the computational domain represents a plane channel of half-height $\delta$ with dimensions $(L_x,L_y,L_z)=(6.4\delta,\,2\delta,\,3.2\delta)$, which is discretized using a finite-volume grid with resolution $(N_x,N_y,N_z)=(320,\,256,\,284)$. The body force is adjusted such that the statistically stationary flow corresponds to $Re_\tau=395$. Uniform grid spacing is employed in the streamwise and spanwise directions, corresponding to wall-unit spacings of $\Delta x^+\approx 7.8$ and $\Delta z^+\approx 4.4$. In the wall-normal direction, grid points are clustered toward the walls using a hyperbolic-tangent stretching, resulting in a near-wall spacing of $\Delta y^+\approx 0.6$ at the first off-wall cell and a gradual increase to $\Delta y^+\approx 5.8$ toward the channel centreline. Periodic boundary conditions are imposed in the streamwise ($x$) and spanwise ($z$) directions, while no-slip conditions are enforced at the two walls located at $y=\pm\delta$. The computational grid is illustrated in Fig.~\ref{fig:grid}(a), where every 8th grid line is shown for clarity. The inset compares the HR grid (blue) with the $4\times$ coarsened LR grid (purple) used as input to the super-resolution network.

The circular pipe configuration consists of a straight circular pipe of radius $R$ and axial length $L/R=15$, with the axis aligned with the streamwise direction. The body force is adjusted such that the statistically stationary flow corresponds to $Re_\tau=360$. Periodic boundary conditions are imposed in the streamwise direction and a no-slip condition is enforced at the pipe wall $r=R$. The cross-section is discretized with an O-grid (butterfly) topology, which avoids the singularity of a polar mesh at the pipe axis at the cost of introducing non-orthogonal cells, particularly in the transition region between the central block and the surrounding annular blocks. The mesh consists of $N_{\mathrm{ax}}=640$ cells in the axial direction and $N_{\mathrm{cs}}\approx 3.3\times10^{4}$ cells per cross-sectional plane, giving axial and azimuthal spacings of $\Delta x^+\approx 8.4$ and $R\Delta\theta^+\approx 5.7$. Cells are clustered towards the wall in the radial direction, with $\Delta r^+\approx 0.7$ at the first off-wall cell increasing to $\Delta r^+\approx 4.8$ near the pipe axis. The grid is shown in Fig.~\ref{fig:grid}(b), together with the corresponding $4\times$ coarsened LR mesh.

To assess generalization to an unseen geometry, an additional DNS is performed for turbulent flow through a straight elliptic pipe with semi-major axis $a$, semi-minor axis $b$, and cross-sectional aspect ratio $a/b=2$. The pipe axis is aligned with the streamwise direction, and the axial length is $L/b=15$. The friction Reynolds number is $Re_{\tau,h}=u_\tau D_h/(2\nu)=412$, where $D_h=4\pi ab/P$ is the hydraulic diameter, $P$ is the cross-sectional perimeter, and $u_\tau$ is calculated from the wall-shear stress averaged over the wall surface and time. Periodic boundary conditions are imposed in the streamwise direction, and a no-slip condition is enforced at the pipe wall. The flow is driven at a constant bulk flow rate using the same body-force procedure as in the other configurations. The cross-section is discretized with an O-grid (butterfly) topology. The mesh consists of $N_{\mathrm{ax}}=512$ cells in the axial direction and $N_{\mathrm{cs}}=55\,296$ cells per cross-sectional plane, giving an axial spacing of $\Delta x^+\approx9.3$. Cells are clustered towards the wall, with a wall-normal spacing of approximately $0.8$ in inner units at the first off-wall cell, increasing to approximately $9$ in the pipe interior. The grid is shown in Fig.~\ref{fig:grid}(c), together with the corresponding $4\times$ coarsened LR mesh.

\begin{figure}
  \centering
  \includegraphics[width=0.95\textwidth]{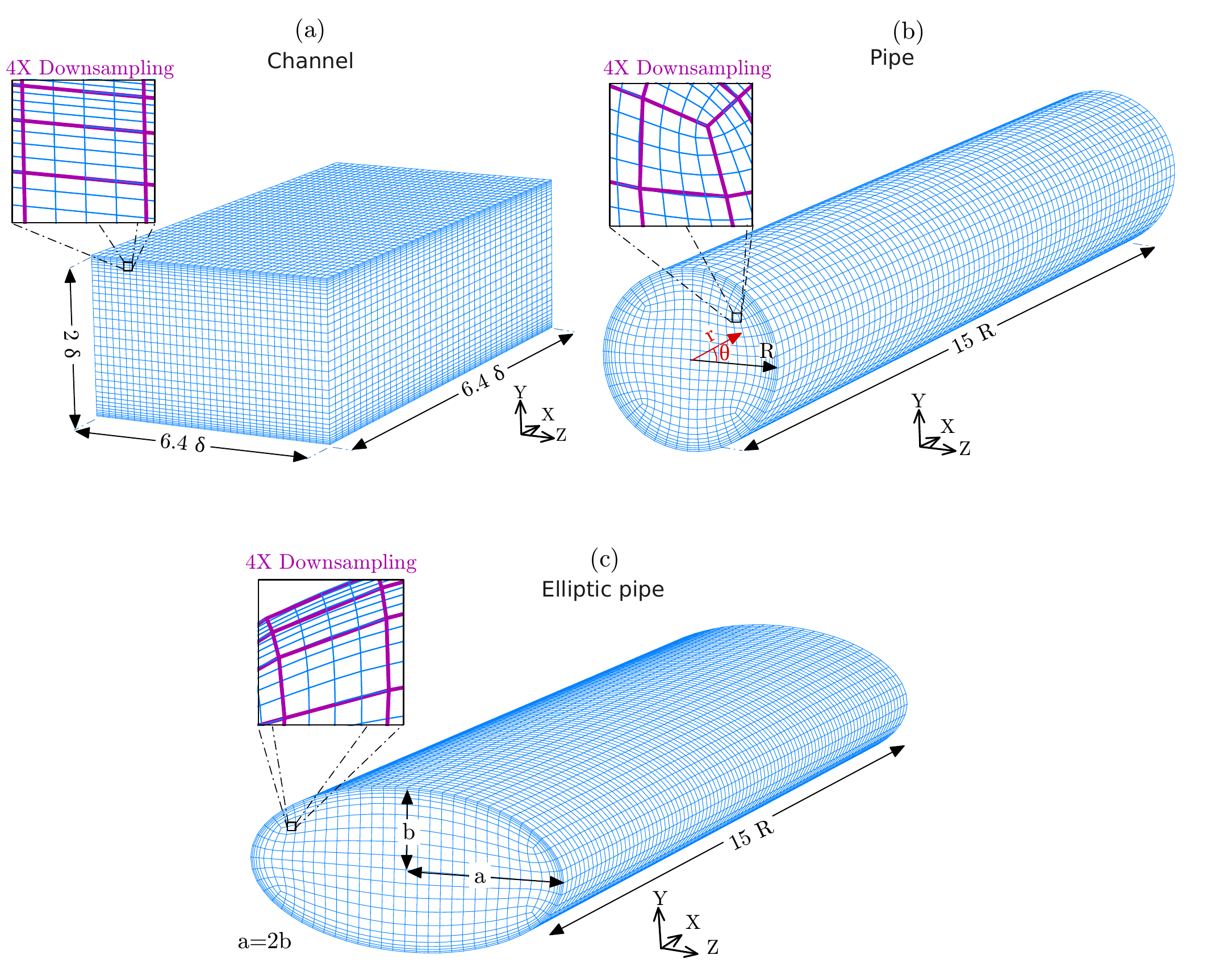}
  \caption{Computational domains and grids for (a) plane channel flow, (b) circular-pipe flow, and (c) elliptic-pipe flow with cross-sectional aspect ratio $a/b=2$. Here, $\delta$ is the channel half-height, $R$ is the circular pipe radius, and $a$ and $b$ are the semi-major and semi-minor axes of the elliptic cross-section, respectively. Every eighth grid line is shown in the main views for clarity. Insets illustrate the HR grid (blue) and the corresponding LR grid (purple), obtained by coarsening by a factor of four in each mesh direction.}
\label{fig:grid}
\end{figure}

The governing equations are solved using the open-source CFD package OpenFOAM\textregistered\ v2112. The diffusion and convection terms are discretized using standard second-order and fourth-order central differencing schemes, respectively. Temporal integration is performed using a second-order upwind Euler scheme, and the maximum Courant--Friedrichs--Lewy number is maintained below $0.2$ throughout the simulations. Velocity--pressure coupling is handled using the Pressure Implicit with Splitting of Operators (PISO) algorithm \citep{issa1986solution}.

To accelerate the development of turbulence, all three simulations are initialized using synthetic divergence-free inflow turbulence \citep{kroger2018generation}. After one flow-through time, once the domain is filled with synthetic turbulent structures, the streamwise boundary condition is switched to periodicity, and the flow is sustained by the body-force driving. The simulations are then continued for an additional ten flow-through times to ensure a statistically stationary state turbulence. From the stationary regime, $200$ three-dimensional snapshots are extracted for training and $100$ for testing in each of the channel and circular-pipe configurations. For the elliptic pipe, $50$ snapshots are extracted exclusively for testing, without retraining or fine-tuning the network. In all three configurations, snapshots are sampled at intervals of $\Delta T^+=100$ in inner units. Hereafter, $\langle \cdot \rangle$ denotes averaging over the test snapshots and the homogeneous directions, namely the streamwise and spanwise directions for the channel and the streamwise and azimuthal directions for the pipe.

To assess the quality of the DNS data used for training and evaluation, we compare first- and second-order turbulence statistics with the reference DNS of \citet{moser1999direct} for the channel and \citet{el2013direct} for the circular pipe. Figure~\ref{fig:valid}(a,c) shows the mean streamwise velocity profiles \(\langle u^+ \rangle\) in inner scaling, which are in close agreement with the reference profiles across the viscous sublayer, buffer layer, and logarithmic region. Figure~\ref{fig:valid}(b,d) shows the Reynolds-stress components \(\langle u_i' u_j' \rangle^+\). The normal stresses reproduce the reference near-wall peak locations and amplitudes, and the Reynolds shear stress exhibits the characteristic negative minimum in the buffer layer, consistent with the reference data for both configurations.

\begin{figure}
  \centering
  \includegraphics[width=1.0\textwidth]{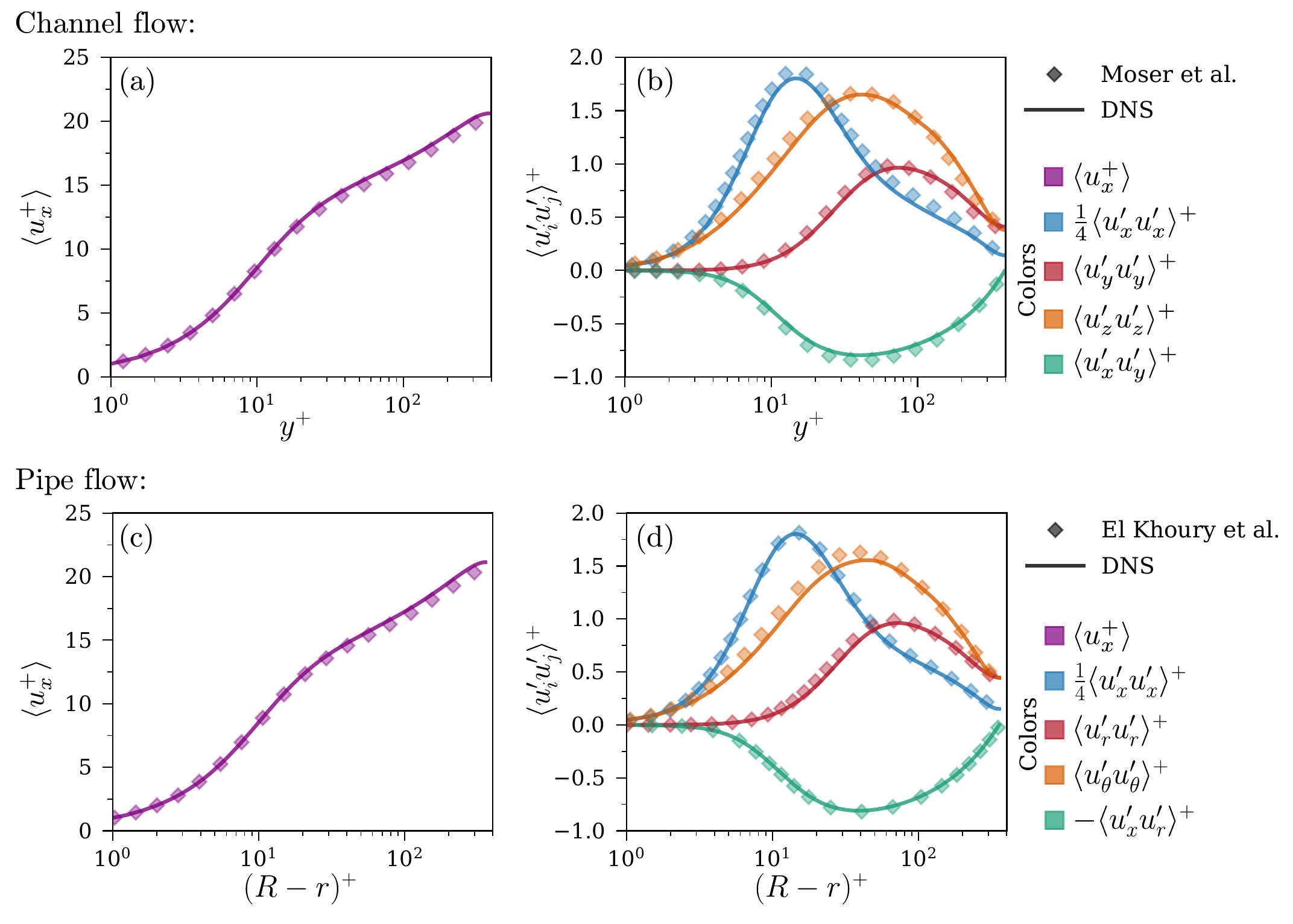}
  \caption{Validation of the present DNS against reference data. (a,\,b) Turbulent channel flow at $Re_\tau \approx 395$ compared with \cite{moser1999direct}. (c,\,d) Turbulent circular-pipe flow at $Re_\tau \approx 360$ compared with \cite{el2013direct}. (a,\,c) Mean streamwise velocity profile in inner units as a function of the wall distance. (b,\,d) Reynolds stress components $\langle u_i' u_j' \rangle^+$ in inner units. Symbols denote the reference data, and solid lines correspond to the present DNS.}
\label{fig:valid}
\end{figure}

The LR inputs are constructed from each DNS snapshot by low-pass filtering the HR DNS velocity field \(\mathbf{u}^{DNS}\) prior to downsampling. Filtering is performed by solving the anisotropic Helmholtz equation
\begin{equation}
\mathbf{u}^{F}-\nabla\cdot\!\left(\mathbf{D}\nabla\mathbf{u}^{F}\right)
=\mathbf{u}^{DNS},
\end{equation}
for the filtered field \(\mathbf{u}^{F}\). This provides a finite-volume-consistent alternative to an explicit Gaussian filter, which is not straightforward to
apply on unstructured meshes. The diffusion tensor is taken as \(\mathbf{D}=\mathrm{diag}(\nu_x,\nu_y,\nu_z)\), with \(\nu_i=\sigma_i^2/2\). The directional filter widths are defined as \(\sigma_i=C_i\,\Delta_i\) (\(i\in\{x,y,z\}\)), where \(\Delta_i\) is a local cell-size measure. In the present work we use \(C_i=5\) in all directions. The LR field \(\mathbf{u}^{LR}\) is then obtained by downsampling \(\mathbf{u}^{F}\) onto a mesh that is \(4\times\) coarser in each coordinate direction.

\subsection{Training}
\label{sec:training}
The network is trained in a supervised manner using paired LR inputs and DNS targets. The velocity fields are normalized by the mean velocity \(u_{\mathrm{in}}\) of the corresponding configuration, \(\tilde{\mathbf{u}}=\mathbf{u}/u_{\mathrm{in}}\). To enable finite-volume operations during training, each OpenFOAM data snapshot, together with the mesh topology, is converted into a PyTorch Geometric data format. All length scales are expressed in inner units during this conversion. The cell-centred fields are stored as node features, while mesh adjacency is represented by the list of cells sharing each face. The geometric quantities required by the finite-volume convolution units (FVCUs), namely cell volumes, face area vectors, and cell and face centroids, are stored as graph attributes. The network is trained by minimizing the objective
\begin{equation}
\mathcal{L}
=
\alpha_{\mathrm{vel}}\,\mathcal{L}_{\mathrm{vel}}
+
\alpha_{\mathrm{grad}}\,\mathcal{L}_{\mathrm{grad}}
+
\alpha_{\mathrm{div}}\,\mathcal{L}_{\mathrm{div}},
\label{eq:loss-total}
\end{equation}
with \(\alpha_{\mathrm{vel}}=0.89\), \(\alpha_{\mathrm{grad}}=0.085\), and \(\alpha_{\mathrm{div}}=0.025\), following \citet{tofighian2024conditional}.

To account for non-uniform cell sizes, all losses are computed as volume-weighted discrete norms using the normalized weight \(\omega_P=V_P/V_\Omega\), where \(V_P\) is the volume of cell \(P\) and \(V_\Omega\) is the total domain volume.

Let \(\tilde{\mathbf{u}}^{SR}(P)\) and \(\tilde{\mathbf{u}}^{DNS}(P)\) denote the predicted SR and target velocities in normalized form. The velocity loss term uses a Charbonnier penalty \citep{charbonnier1994two}, a smooth \(L_1\) loss,
\begin{equation}
\mathcal{L}_{\mathrm{vel}}
=
\sum_{P}
\omega_P\,
\sum_{i=1}^{3}
\sqrt{\left(\tilde{u}_i^{SR}(P)-\tilde{u}_i^{DNS}(P)\right)^2+\varepsilon^2},
\qquad \varepsilon=10^{-3}.
\label{eq:loss-vel}
\end{equation}
To promote accurate small-scale structure, we additionally match velocity gradients with an \(L_2\) loss,
\begin{equation}
\mathcal{L}_{\mathrm{grad}}
=
\sum_{P}
\omega_P\,
\left\|
\nabla\tilde{\mathbf{u}}^{SR}(P)
-
\nabla\tilde{\mathbf{u}}^{DNS}(P)
\right\|_2^2,
\label{eq:loss-grad}
\end{equation}
and the divergence loss penalizes violations of incompressibility,
\begin{equation}
\mathcal{L}_{\mathrm{div}}
=
\sum_{P}
\omega_P\,
\left(
\nabla\!\cdot\tilde{\mathbf{u}}^{SR}(P)
\right)^2,
\label{eq:loss-div}
\end{equation}

Training was performed using the Adam optimizer with a learning rate of \(10^{-4}\). The model was trained for approximately \(15\) hours and training was stopped once the validation losses reached a plateau. All runs were carried out on a single NVIDIA H200 GPU of the HoreKa supercomputer at the Karlsruhe Institute of Technology (KIT).

\section{Results}\label{sec:results}

\subsection{Turbulent channel flow}

Figure~\ref{fig:channel_velocity} presents instantaneous velocity fluctuations in the $x$–$y$ plane of the turbulent channel flow, comparing the three components $u_x'/u_{\mathrm{in}}$, $u_y'/u_{\mathrm{in}}$, and $u_z'/u_{\mathrm{in}}$ reconstructed by the proposed FVCNN against the DNS reference and the LR input. It is evident that the LR field retains the large-scale distribution of velocity fluctuations but loses much of the fine-scale structure due to the low-pass filtering and $4\times$ downsampling. This results in spatially diffuse contours, with weaker local variations and less distinct boundaries between positive and negative fluctuation regions. The FVCNN reconstruction shows close agreement with the DNS field across all three velocity components, recovering considerably finer spatial structure than that present in the LR input. For the streamwise component, the elongated regions of positive and negative fluctuations near the walls are reproduced with similar locations and shapes to those in the DNS. The improvement is also apparent in the wall-normal and spanwise components, where the localized alternating fluctuations, largely smoothed out in the LR field, are recovered by the network. Both the fine near-wall structures and the broader patterns in the channel interior are captured.

\begin{figure}
  \centering
  \includegraphics[width=1.0\textwidth]{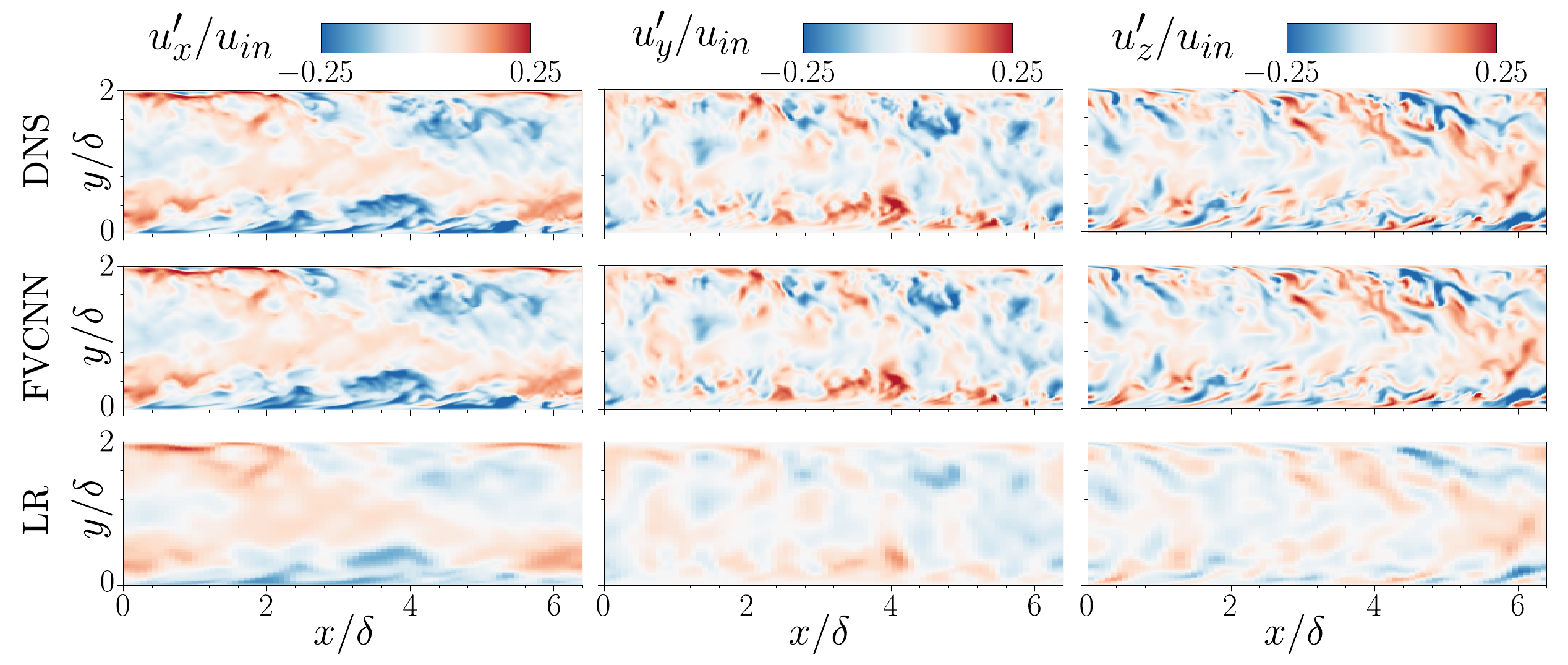}
  \caption{Instantaneous velocity fluctuations in the $x$–$y$ plane of turbulent channel flow at $Re_\tau=395$. Rows show DNS (top), FVCNN reconstruction (middle), and LR input (bottom). Columns correspond to the streamwise, wall-normal, and spanwise velocity fluctuations, respectively, normalized by $u_{\mathrm{in}}$. Spatial coordinates are normalized by the channel half-height $\delta$.}
  \label{fig:channel_velocity}
\end{figure}

Fig.~\ref{fig:R_channel}(a) compares wall-normal profiles of the Reynolds stresses in inner scaling, $\langle u_i'u_j'\rangle^+$, as functions of $y^+$. The LR field exhibits the expected attenuation of turbulent intensity and suppressed momentum transfer due to the low-pass filtering and downsampling. In contrast, the FVCNN reconstruction closely reproduces the DNS profile in both peak location and peak magnitude for the streamwise component $\langle u_x'u_x'\rangle^+$. The near-wall maximum is captured with the correct sharp rise from the viscous sublayer and a peak in the buffer layer ($y^+\approx15$). The FVCNN also recovers the overall trends of the wall-normal and spanwise components, $\langle u_y'u_y'\rangle^+$ and $\langle u_z'u_z'\rangle^+$. However, it shows a small underestimation of these two components around their peak region. The deficit decreases gradually with increasing $y^+$ and becomes negligible towards the channel centre. The Reynolds shear stress $\langle u_x'u_y'\rangle^+$ is recovered with the correct sign and a distinct negative minimum, demonstrating that the FVCNN reconstruction restores the turbulent momentum transfer that is strongly damped in the LR input.

To further assess the Reynolds-stress predictions, the state of turbulence anisotropy is visualized using a barycentric map \citep{banerjee2007presentation} in Fig.~\ref{fig:R_channel}(b). The anisotropy tensor corresponding to the Reynolds-stress tensor $R_{ij}=\langle u_i'u_j'\rangle$ is mapped onto the barycentric triangle whose vertices correspond to the one-component (1C), two-component (2C), and three-component isotropic (3C) limiting states. The trajectories are colored by wall distance $y^+$, so that the evolution of the anisotropy state from the viscous sublayer to the channel center can be traced continuously. The DNS trajectory starts near the two-component boundary in the near-wall limit, where wall-normal fluctuations are suppressed, and subsequently migrates toward the isotropic vertex with increasing $y^+$, consistent with the well-known return to isotropy in the channel core \citep{pope2001turbulent}. However, the LR field remains biased towards a more anisotropic state across most of the wall-normal range. The FVCNN reconstruction closely follows the DNS trajectory, with only a slight delay in the approach toward isotropy as $y^+$ increases. This deviation is small, and the overall evolution of the anisotropy state is well reproduced.

\begin{figure}
  \centering
  \includegraphics[width=1.0\textwidth]{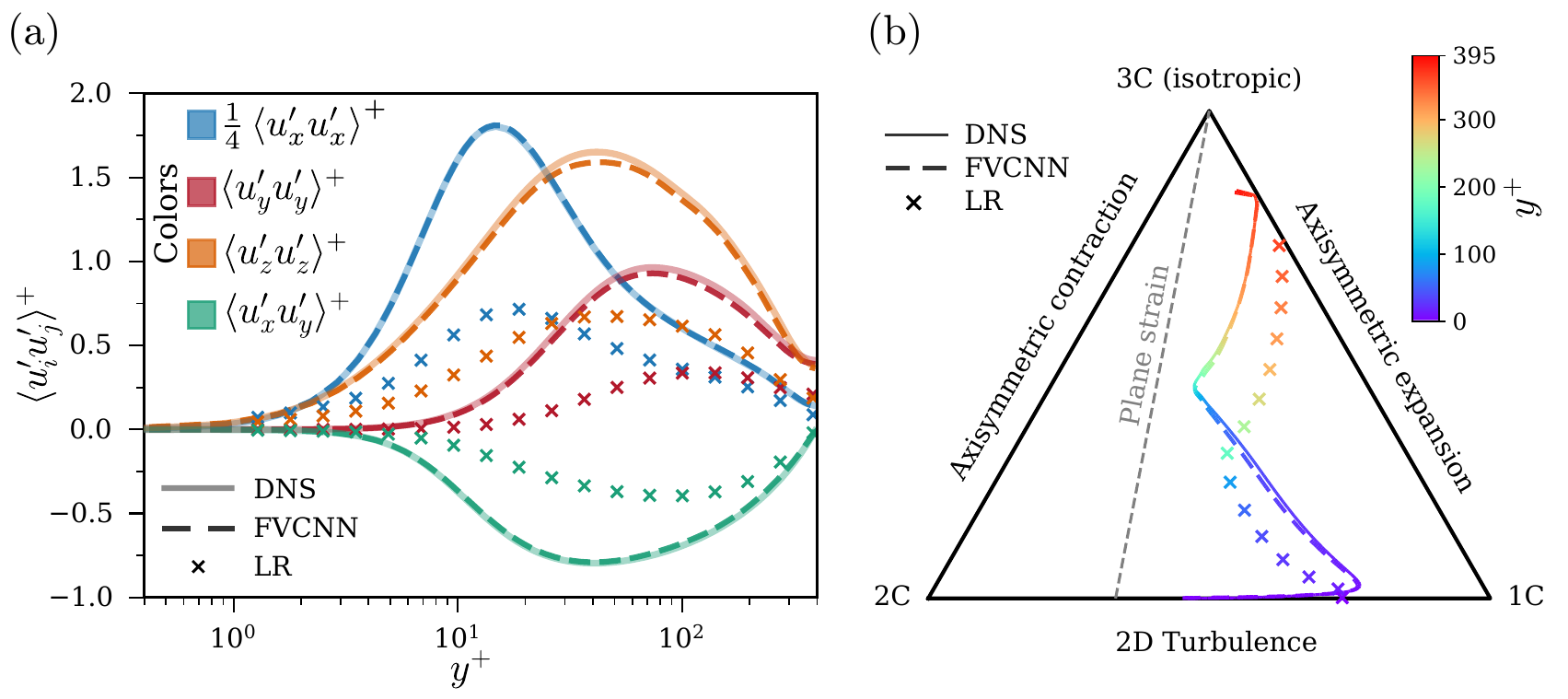}
  \caption{Assessment of Reynolds-stress reconstruction in turbulent channel flow at $Re_\tau=395$: (a) wall-normal profiles of Reynolds stresses in inner scaling and (b) barycentric map of Reynolds-stress anisotropy \citep{banerjee2007presentation}. The streamwise normal stress in (a) is divided by four for clarity. Solid lines denote DNS, dashed lines denote FVCNN, and crosses denote LR data. In (b), colors indicate wall distance $y^+$, and the vertices represent the one-component (1C), two-component (2C), and three-component isotropic (3C) limiting states.}
\label{fig:R_channel}
\end{figure}

To assess the ability of the FVCNN to reproduce the scale-by-scale distribution of fluctuation energy, Fig.~\ref{fig:channel_FFT} compares the spanwise energy spectra of the three velocity components at $y^+=15$ and $y^+=60$ against the DNS reference. As expected, the LR energy spectra deviate from the DNS spectra due to the filtering process, especially at high wavenumbers. This deviation is already apparent below the Nyquist limit of the LR mesh and becomes more pronounced as $k_z^+$ increases. It is seen from Fig.~\ref{fig:channel_FFT} that the FVCNN spectra closely match the DNS results at low and intermediate wavenumbers for all three components and both wall distances. The wall-normal spectrum at $y^+=15$ is also recovered with a peak location and magnitude close to those of the DNS, whereas the LR field substantially under-predicts this peak.

Beyond the LR Nyquist limit, indicated by the shaded regions, the FVCNN recovers much of the energy removed from the input field. This observation is consistent with the velocity-field comparison in Fig.~\ref{fig:channel_velocity}, where the FVCNN reconstructs fine-scale structures that are absent from the LR input. However, at high wavenumbers, the reconstructed spectra generally under-predict the energy compared to the DNS, indicating that the smallest-scale fluctuations are only partially recovered.

\begin{figure}
  \centering
  \includegraphics[width=0.7\textwidth]{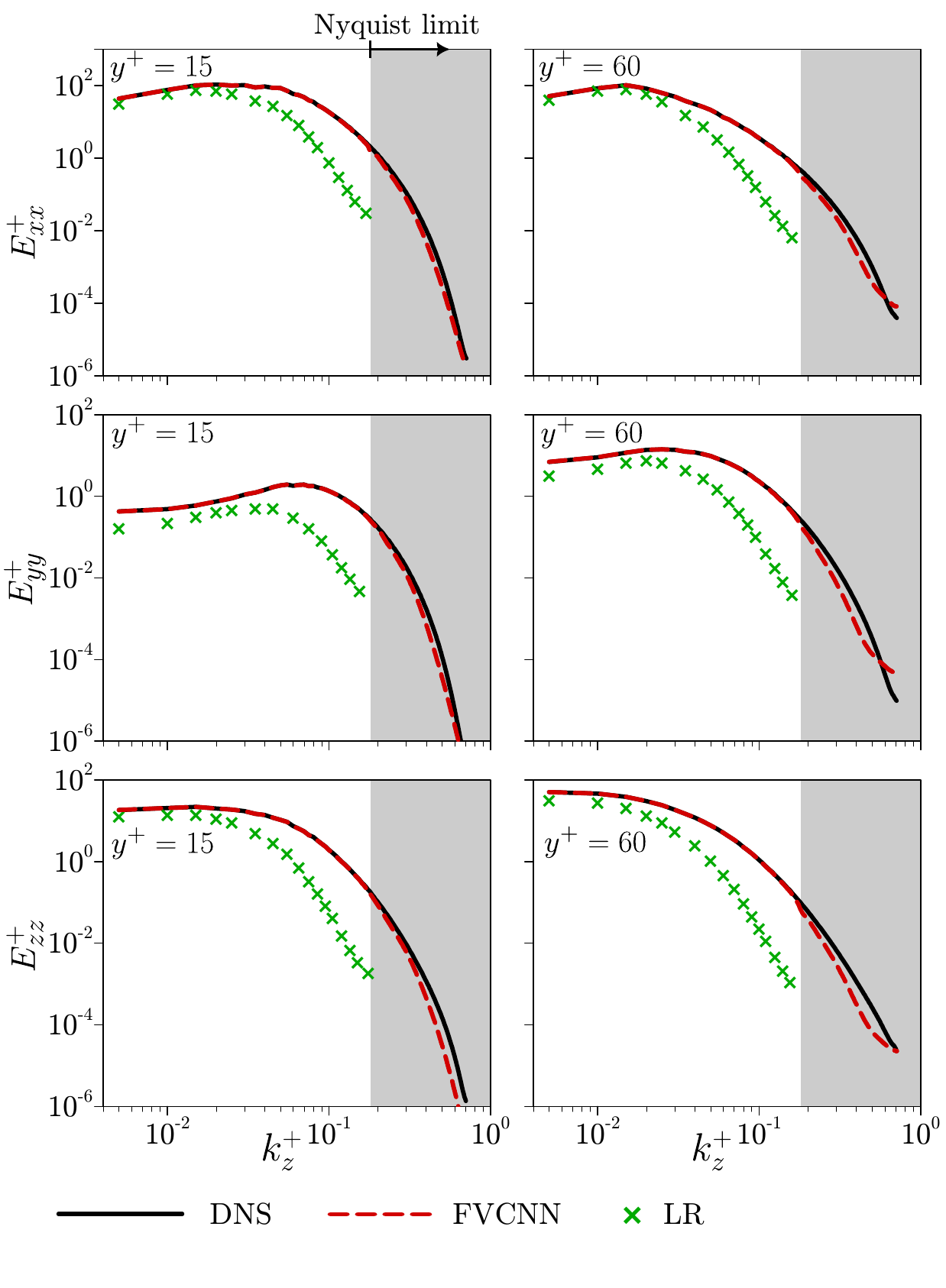}
  \caption{Comparison of one-dimensional spanwise energy spectra in turbulent channel flow at $Re_\tau=395$. Spectra are shown in inner scaling at $y^+=15$ (left column) and $y^+=60$ (right column). Rows correspond to the streamwise, wall-normal, and spanwise velocity components, respectively. Black solid lines denote DNS, red dashed lines denote FVCNN, and green crosses denote LR data. Shaded regions indicate wavenumbers beyond the Nyquist limit of the LR mesh.}
  \label{fig:channel_FFT}
\end{figure}

\subsubsection{Effect of wall-proximity encoding}
To assess the importance of the Van Driest damping function as a conditional input channel, we perform an ablation study \citep{shamooni2025super} at inference time by replacing the wall-proximity input with a spatially uniform value ($D=1$) or setting it to zero ($D=0$). Figure~\ref{fig:vanD} shows the resulting barycentric maps of total Reynolds-stress anisotropy. When the wall-proximity channel is set to zero ($D=0$), the anisotropy trajectory shifts toward the 2C vertex. This is physically consistent, as $D=0$ corresponds to the wall limit, and the network responds with stronger suppression of wall-normal fluctuations, driving the reconstructed turbulence toward a two-component state across the entire wall-normal range. Conversely, supplying a uniform value ($D=1$) shifts the near-wall trajectory toward the isotropic vertex. Since $D\approx 1$ characterizes the region far from the wall where turbulence is approximately isotropic, the network interprets this input as an absence of wall influence and accordingly reduces the anisotropy of the reconstructed field, even at locations where the true turbulence is strongly anisotropic. Away from the wall ($y^+\gtrsim 100$), where the Van Driest function is itself close to unity, the $D=1$ trajectory converges toward the DNS, as expected. The ablation results show that the network uses the Van Driest input in a physically consistent manner, modulating the reconstructed anisotropy according to wall proximity. The degradation observed in both modified cases indicates that wall-distance information cannot be recovered from the velocity field alone and its inclusion is essential for an accurate reconstruction of the anisotropy state in wall-bounded turbulence.
\begin{figure}
  \centering
  \includegraphics[width=0.5\textwidth]{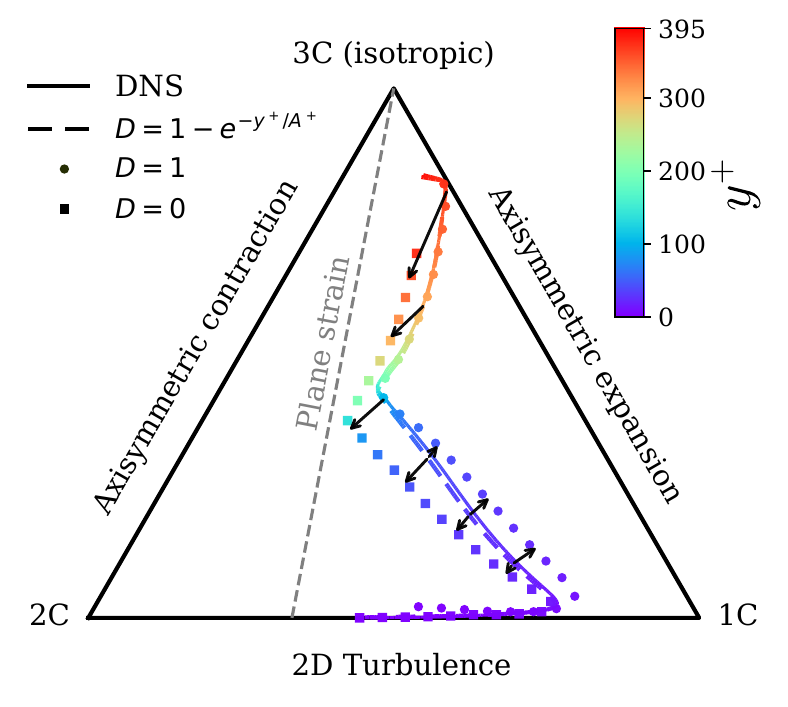}
  \caption{Barycentric map of total Reynolds-stress anisotropy coloured by wall distance $y^+$, comparing the DNS with FVCNN reconstructions using three wall-proximity inputs: the Van Driest damping function $D = 1 - e^{-y^+/A^+}$, a uniform value $D = 1$, and a zero value $D = 0$. The vertices 1C, 2C, and 3C correspond to the one-, two-, and three-component (isotropic) limiting states, respectively. Arrows show the change in anisotropy at the same $y^+$ when the input is switched from $D = 1 - e^{-y^+/A^+}$ to $D=0$ or $D=1$.}
\label{fig:vanD}
\end{figure}

\subsection{Turbulent circular-pipe flow}\label{sec:pipe}
The FVCNN is further evaluated for turbulent circular-pipe flow at $Re_\tau=360$. Unlike the channel mesh, which is orthogonal but non-uniform, the butterfly topology of the pipe cross-section introduces both non-uniform spacing and non-orthogonal cells. This configuration therefore extends the assessment to a more complex mesh geometry.

Fig.~\ref{fig:R_pipe}~(a) compares the Reynolds-stress profiles in inner scaling as functions of the wall distance $(R-r)^+$. The stresses are expressed in cylindrical coordinates, where $u_x$, $u_r$, and $u_\theta$ denote the axial, radial, and azimuthal velocity components, respectively. Similar to the channel case, the LR field exhibits considerable attenuation of the normal stresses and the magnitude of the Reynolds shear stress due to the low-pass filtering and downsampling. In contrast, the FVCNN closely follows the DNS profiles across the wall-normal range. The axial normal stress $\langle u_x'u_x'\rangle^+$ is recovered with a near-wall peak close to the DNS in both location and magnitude. The radial and azimuthal normal stresses are also well reproduced, although small deficits remain around their peak regions. The shear stress, plotted as $-\langle u_x'u_r'\rangle^+$, closely matches the DNS, including its negative minimum and subsequent approach to zero at the pipe axis. This agreement indicates that the reconstruction recovers both the velocity fluctuation intensities and the turbulent momentum transfer substantially attenuated in the LR input.

The corresponding barycentric map is shown in Fig.~\ref{fig:R_pipe}~(b), with trajectories colored by wall distance $(R-r)^+$. The DNS trajectory lies near the two-component boundary in the near-wall limit and initially moves towards the one-component vertex, reflecting the dominance of axial velocity fluctuations in the buffer layer. This behavior is consistent with the pronounced axial normal-stress peak near $(R-r)^+\approx15$ in Fig.~\ref{fig:R_pipe}~(a). With increasing distance from the wall, the trajectory turns towards the isotropic vertex. The LR trajectory deviates from the DNS evolution, particularly at intermediate and larger wall distances, where it remains farther from the isotropic vertex. In contrast, the FVCNN closely follows the DNS trajectory over the wall-normal range, with only minor differences. Similar to the channel case, the FVCNN accurately reproduces both the Reynolds-stress magnitudes and their anisotropy on the non-orthogonal pipe mesh.
\begin{figure}
  \centering
  \includegraphics[width=1.00\textwidth]{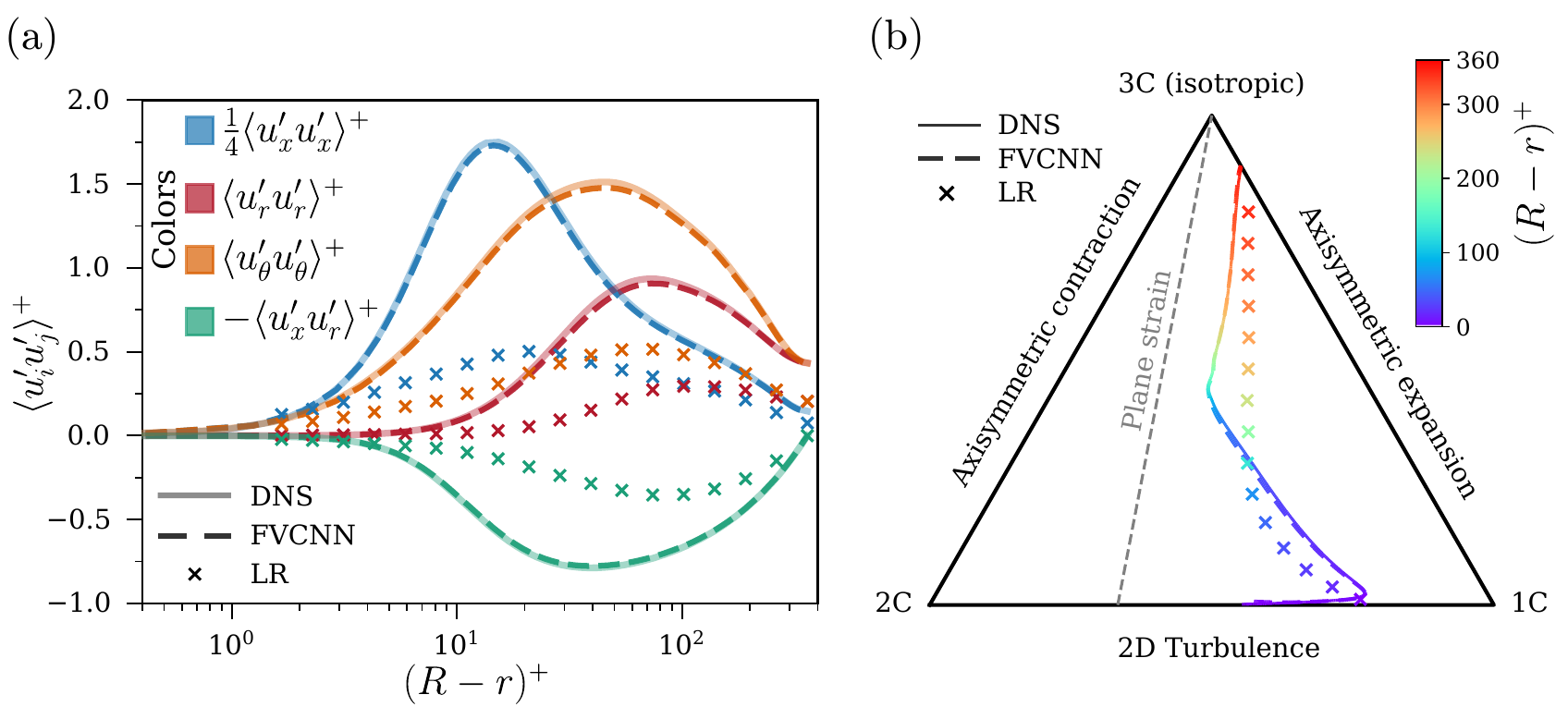}
  \caption{Assessment of Reynolds-stress reconstruction in turbulent circular-pipe flow at \(Re_\tau=360\): (a) Reynolds-stress profiles in inner scaling versus wall distance \((R-r)^+\), comparing DNS, FVCNN, and LR data; (b) barycentric map of Reynolds-stress anisotropy, colored by \((R-r)^+\). The axial normal stress in (a) is divided by four for clarity. The vertices in (b) represent the one-component (1C), two-component (2C), and three-component isotropic (3C) limiting states.}
\label{fig:R_pipe}
\end{figure}

Fig.~\ref{fig:pipe_FFT} compares the one-dimensional axial energy spectra $E_{xx}^+$ and $E_{rr}^+$ at wall distances $(R-r)^+=15$ and $60$. The LR spectra exhibit an energy deficit for both velocity components at both wall distances. This deficit increases with wavenumber and is already substantial below the Nyquist limit of the LR mesh, reflecting the attenuation introduced by filtering. In contrast, the FVCNN closely reproduces the DNS spectra at low and intermediate wavenumbers. The agreement continues beyond the LR Nyquist limit, indicated by the shaded regions, where the reconstruction recovers much of the energy absent from the coarse input.

Differences between the FVCNN and DNS become apparent at higher wavenumbers. The reconstructed spectra first under-predict the DNS energy over part of this range and then exceed it at the highest plotted wavenumbers. The latter deviation is more evident at $(R-r)^+=60$, where the DNS spectra continue to decrease sharply while the reconstructed spectra exhibit a slower decay. Thus, despite the close agreement over most of the plotted range, the energy distribution at the smallest reconstructed scales is not fully reproduced.

\begin{figure}
  \centering
  \includegraphics[width=0.7\textwidth]{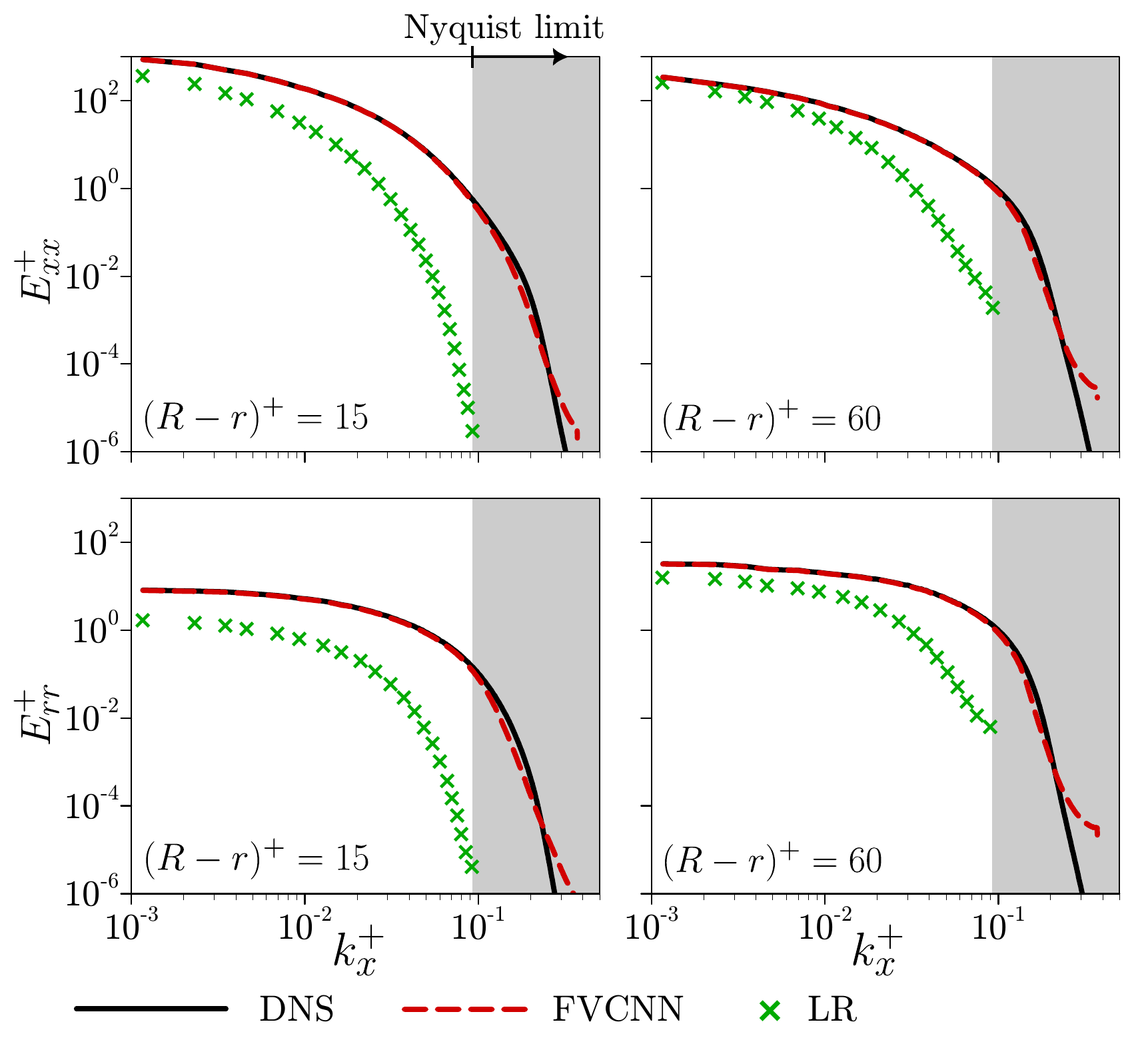}
  \caption{Comparison of one-dimensional axial energy spectra in turbulent circular-pipe flow at $Re_\tau=360$. Spectra are shown in inner scaling at $(R-r)^+=15$ (left column) and $(R-r)^+=60$ (right column). Rows correspond to the axial and radial velocity components, respectively. Black solid lines denote DNS, red dashed lines denote FVCNN, and green crosses denote LR data. Shaded regions indicate wavenumbers beyond the Nyquist limit of the LR mesh.}
  \label{fig:pipe_FFT}
\end{figure}

\subsection{Evaluation on an unseen geometry: turbulent elliptic-pipe flow}\label{sec:elliptic} 

To assess generalization to an unseen geometry, the FVCNN is evaluated on turbulent flow through an elliptic pipe with cross-sectional aspect ratio $a/b=2$. This configuration extends the assessment to a non-circular cross-section that was not included in the training data. Figure~\ref{fig:ellip}~(a,b) compares instantaneous cross-sectional distributions of the normalized axial velocity, $u_x/u_{\mathrm{in}}$, and axial vorticity, $\omega_x b/u_{\mathrm{in}}$, obtained from DNS, the FVCNN reconstruction, and the LR input. The LR velocity field preserves the approximate locations of the high- and low-speed regions but smooths out the finer spatial variations. The FVCNN recovers much of this missing detail, producing sharper boundaries between these regions and restoring smaller structures visible in the DNS. The improvement is particularly apparent in the axial vorticity, for which downsampling strongly attenuate the localized regions of positive and negative rotation. The reconstructed field recovers many of these structures, including the closely spaced patches near the wall. Nevertheless, some differences remain in the shapes and intensities of individual vorticity patches.

The PDFs at the marked location, situated 60 wall units from the wall, provide a complementary statistical assessment (Fig.~\ref{fig:ellip}~(c,d)). The LR axial-velocity distribution is narrower than the DNS distribution, with reduced tails, whereas the FVCNN closely reproduces its overall shape and spread. The contrast is more pronounced for axial vorticity, whose LR PDF is sharply concentrated around zero. Since vorticity involves spatial derivatives of velocity, it places greater emphasis on small-scale motions and is therefore more sensitive to the loss of fine-scale content through downsampling. The FVCNN recovers the spread of the DNS vorticity distribution and closely follows its shape. However, the reconstructed vorticity PDF slightly exceeds the DNS PDF in portions of the tails, indicating a modest overprediction of the probability of extreme axial-vorticity events at this location. 
\begin{figure}
  \centering
  \includegraphics[width=0.8\textwidth]{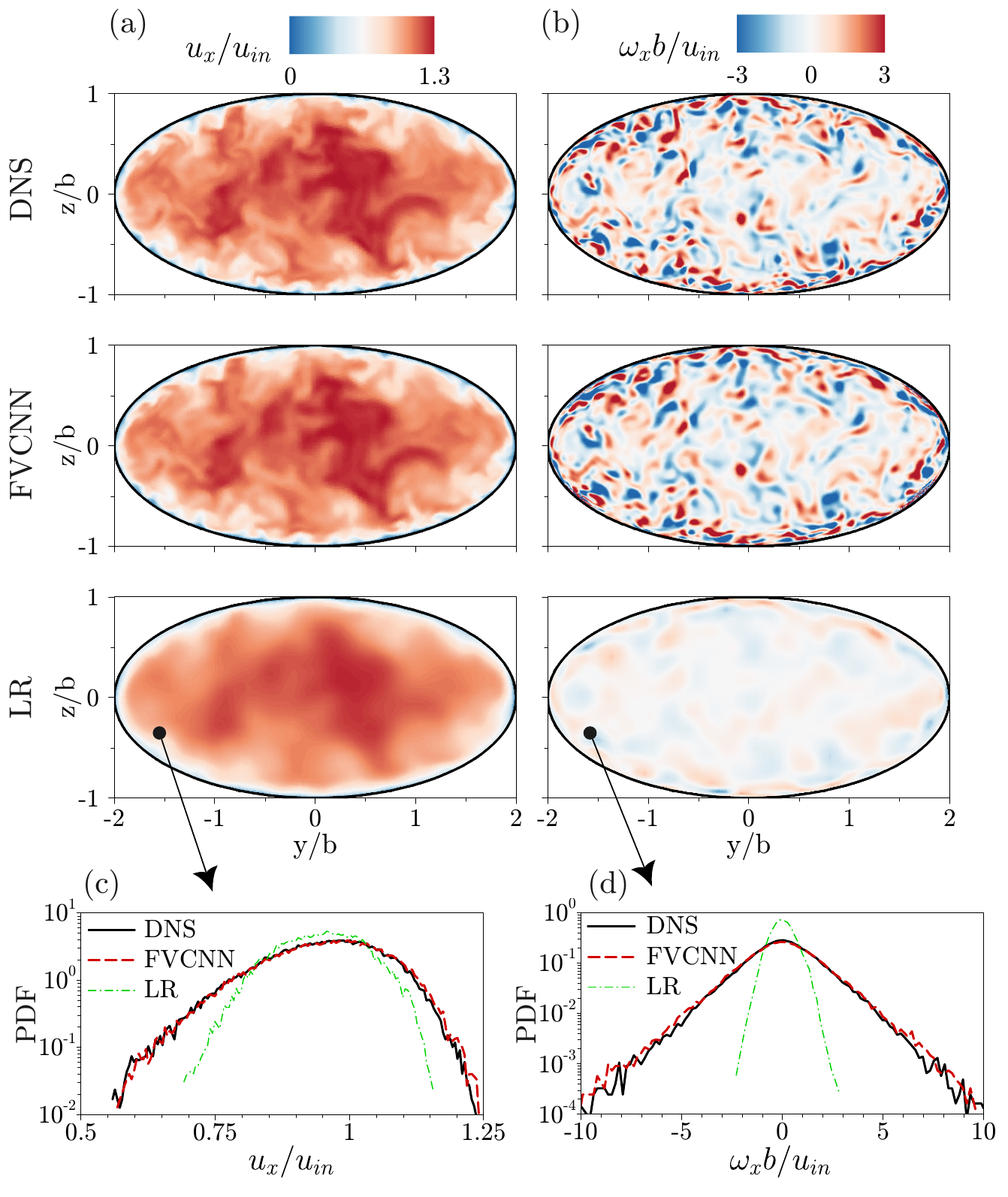}
  \caption{Instantaneous fields and local probability density functions (PDFs) in turbulent elliptic-pipe flow with aspect ratio $a/b=2$. Cross-sectional distributions of (a) normalized axial velocity, $u_x/u_{\mathrm{in}}$, and (b) normalized axial vorticity, $\omega_x b/u_{\mathrm{in}}$, are shown for DNS (top), the FVCNN reconstruction (middle), and the LR input (bottom). Coordinates are normalized by the semi-minor axis $b$, and $u_{\mathrm{in}}$ denotes the bulk velocity. Panels (c) and (d) compare the PDFs of the corresponding quantities at the location marked by the black dots, situated 60 wall units from the wall.}
  \label{fig:ellip}
\end{figure}

The spectral comparison for the elliptic pipe is presented in Fig.~\ref{fig:ellip_FFT}, using the same location as the PDFs, marked in Fig.~\ref{fig:ellip} and situated 60 wall units from the wall. The LR spectra increasingly underestimate the DNS energy as the wavenumber increases, with substantial attenuation already below the Nyquist limit of the LR mesh. The FVCNN closely follows the DNS spectra at low and intermediate wavenumbers for all three velocity components. It also recovers energy beyond the LR Nyquist limit and captures the steep decay of the DNS spectra.

However, the FVCNN overpredicts the energy around and above the LR Nyquist limit, particularly for the two transverse velocity components. This excess is qualitatively consistent with the slight overestimation of the axial-vorticity PDF tails in Fig.~\ref{fig:ellip}(d), suggesting that some of the reconstructed small-scale fluctuations are too strong. The excess energy near the LR Nyquist limit is not observed in the channel and circular pipe results shown in Figs.~\ref{fig:channel_FFT} and \ref{fig:pipe_FFT}, respectively, where the FVCNN follows the DNS more closely in this wavenumber range. Therefore, The FVCNN reproduces the main flow features and local statistics in the unseen elliptic pipe, although the reconstruction of these smaller scales is less accurate than in the geometries represented in the training data.

\begin{figure}
  \centering
  \includegraphics[width=1.0\textwidth]{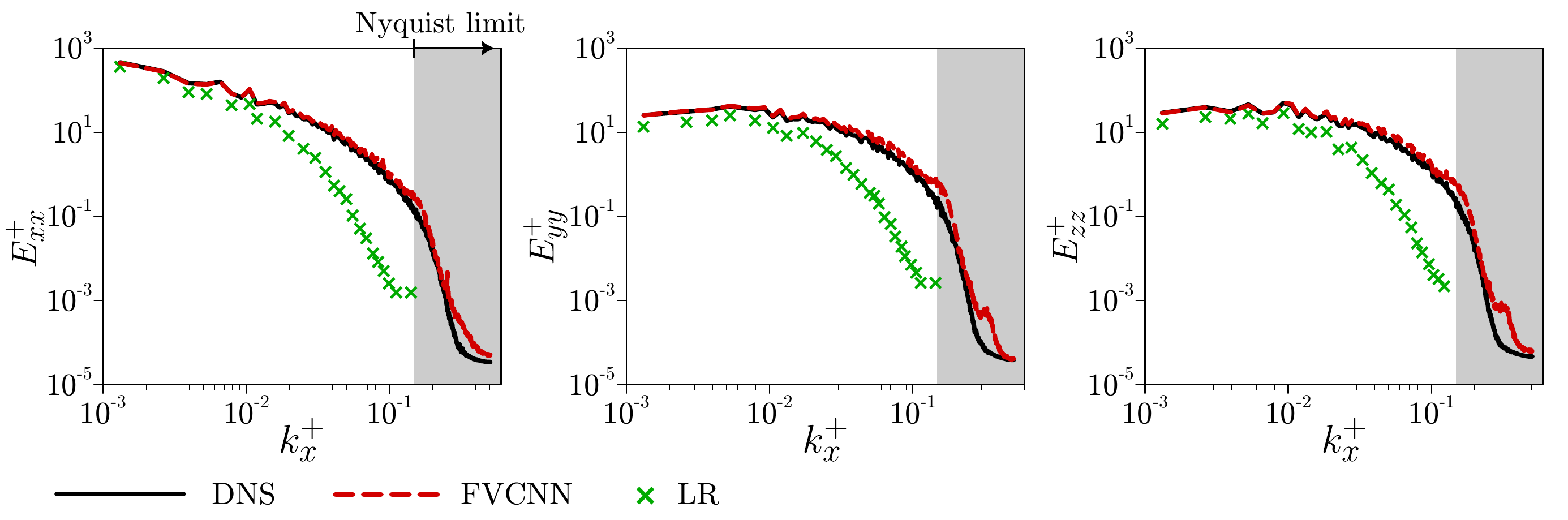}
  \caption{One-dimensional axial energy spectra in turbulent elliptic-pipe flow at the location marked in Fig.~\ref{fig:ellip}, situated 60 wall units from the wall. From left to right, the panels show $E_{xx}^+$, $E_{yy}^+$, and $E_{zz}^+$ as functions of the axial wavenumber $k_x^+$, with all quantities expressed in inner units. Black solid lines denote DNS, red dashed lines denote FVCNN, and green crosses denote LR data. Shaded regions indicate wavenumbers beyond the Nyquist limit of the LR mesh.}
  \label{fig:ellip_FFT}
\end{figure}

\section{Conclusion}\label{sec:Conclusion}
This work extends convolutional neural networks to general 3D finite-volume meshes, removing the restriction to uniform, structured grids that has limited previous CNN-based approaches.

The central idea is to construct convolutional layers using finite-volume operators that remain well-defined on non-uniform and unstructured meshes. This construction builds on an exact decomposition of a 3D convolution as a linear combination of successive axis-wise convolutions. This decomposition is exact for any set of three linearly independent one-dimensional filters forming a basis. We choose the identity, gradient, and averaging operators as the basis, because each can be evaluated on arbitrary polyhedral cells using standard finite-volume discretization schemes and reduces to its finite-difference counterpart on a uniform Cartesian grid. The resulting finite-volume convolution unit (FVCU) serves as a drop-in replacement for the standard convolution layer, allowing established deep-learning architectures to operate directly on finite-volume data. In this work, FVCU is embedded in a residual in residual dense neural network for super-resolution reconstruction of turbulent flows.

The proposed model was trained to reconstruct 3D velocity fields from $4\times$ coarsened inputs in turbulent channel flow at $Re_\tau=395$, on a mesh that is non-uniform in the wall-normal direction, and in circular-pipe flow at $Re_\tau=360$, on a non-orthogonal O-grid mesh. Because standard convolutions require uniform grid spacing, a conventional 3D CNN cannot be applied directly to either configuration. In both configurations, FVCNN accurately reproduced the Reynolds-stress profiles and their anisotropy, as well as the energy spectra, and recovered a substantial part of the energy beyond the Nyquist limit of the LR mesh. Additionally, inspired by the role of the Van Driest damping function in near-wall subgrid-scale modelling in large-eddy simulations, this function was introduced as a conditional input to provide the model with wall-proximity information. This information is required because both the size of turbulent structures and the anisotropy of the Reynolds stresses vary strongly with wall distance. An ablation study confirmed that the network interprets this input in a physically consistent manner.

A further advantage of the present formulation is its transferability across geometries. Although the elliptic-pipe geometry with aspect ratio $a/b=2$ was not included in the training data, the network generalized to this configuration at inference without retraining or fine-tuning. The model reproduced the main flow features, the local velocity and vorticity statistics, and the energy spectra at low and intermediate wavenumbers. However, the smallest scales were reconstructed less accurately than in the training geometries, with a modest overprediction of the energy around and above the Nyquist limit of the LR mesh.

Several directions remain open for future work. The reconstruction of the smallest scales, particularly in unseen geometries, may be improved by training on a broader range of geometries, mesh topologies, and Reynolds numbers. To improve the capabilities of the model on strongly non-orthogonal grids, one option is to incorporate non-orthogonal corrections into the finite-volume basis operators, following standard practice in finite-volume CFD methods. Alternatively, training the network on meshes with strong non-orthogonality may enable the model to learn the necessary corrections implicitly, without modifying the basis operators. Since the FVCU is formulated for arbitrary polyhedral cells, its assessment on more complex geometries and on meshes with mixed cell types, such as tetrahedral, prismatic, and polyhedral cells, is a natural next step. Moreover, because the FVCU is based on standard finite-volume discretizations that are already implemented in established CFD codes, the integration of the method into existing solvers is facilitated. More broadly, although the present work focused on super-resolution reconstruction, the FVCU provides a general framework for applying deep learning on finite-volume meshes. Its application to other tasks, such as turbulence modelling, flow prediction, or data assimilation on unstructured grids, is a promising direction for future research.

\backsection[Acknowledgements]{This work was performed on the HoreKa supercomputer funded by the Ministry of Science, Research and the Arts Baden-Württemberg and by the Federal Ministry of Education and Research. The authors acknowledge support by the state of Baden-Württemberg through bwHPC.} 

\backsection[Declaration of interests]{The authors report no conflict of interest.}

\backsection[Author contributions]{Hesam Tofighian: Conceptualization (lead); Formal analysis (lead); Investigation (lead); Methodology (lead); Software (lead); Visualization (lead); Writing – original draft (lead). Jordan A. Denev: Resources (equal); Supervision (equal); Writing – review \& editing (equal). Nikolai Kornev: Project administration (lead); Resources (equal); Supervision (equal); Writing – review \& editing (equal).}

\backsection[Data availability]{The data that support the findings of this study are available from the corresponding author upon reasonable request.}

\appendix
\section{}\label{appA}
\begin{proposition}
\label{thm:axiswise-1D}
Let $\{b_1,b_2,b_3\}$ be a basis of $\mathbb{R}^3$ and set $\mathcal{K}:=\mathbb{R}^{3\times3\times3}$. For a discrete field $X:\mathbb{Z}^3\to\mathbb{R}$, denote by $b\star_x X$, $b\star_y X$, $b\star_z X$ the length-$3$ 1D convolutions along the $x$-, $y$-, and $z$-axes, respectively, and by $\star$ the standard $3$D convolution. 

Then for every $K\in\mathcal{K}$ there exist unique coefficients $c_{pqr}\in\mathbb{R}$, $p,q,r\in\{1,2,3\}$, such that for all $X$,
\begin{equation}
\label{eq:main}
(K\star X) \;=\; \sum_{p,q,r=1}^{3} c_{pqr}\,\Bigl(b_r\star_z\bigl(b_q\star_y(b_p\star_x X)\bigr)\Bigr).
\end{equation}
\end{proposition}

\begin{proof}
Let $B=[\,b_1\ b_2\ b_3\,]\in\mathbb{R}^{3\times 3}$, which is invertible since $\{b_i\}$ is a basis. Hence the Kronecker product $T:=B\otimes B\otimes B\in\mathbb{R}^{27\times 27}$ is invertible. Therefore the $27$ tensors \[E_{pqr}:=b_p\otimes b_q\otimes b_r,\qquad p,q,r\in\{1,2,3\},\] form a basis of $\mathcal{K}$, so each $K\in\mathcal{K}$ admits a unique expansion
\[K=\sum_{p,q,r} c_{pqr}E_{pqr}.\]
Next, each $E_{pqr}$ is separable across the three coordinate axes, and therefore its $3$D convolution with any field $X$ factorizes as successive 1D convolutions:
\begin{equation}
\label{eq:separable-identity}
(E_{pqr}\star X) \;=\; b_r\star_z\bigl(b_q\star_y(b_p\star_x X)\bigr).
\end{equation}
Using the linearity of convolution in the kernel, we obtain
\[
K\star X
=\sum_{p,q,r} c_{pqr}\,(E_{pqr}\star X)
=
\sum_{p,q,r} c_{pqr}\,\Bigl(b_r\star_z\bigl(b_q\star_y(b_p\star_x X)\bigr)\Bigr),
\]
which is \eqref{eq:main}. Uniqueness of the coefficients $\{c_{pqr}\}$ follows from uniqueness of the basis expansion of $K$.
\end{proof}

\bibliographystyle{jfm}
\bibliography{jfm}
\end{document}